\documentclass[11pt]{amsart}
\usepackage[utf8]{inputenc}
\usepackage{color}
\usepackage{mathtools}
\usepackage{amsbsy}
\usepackage{amstext}
\usepackage{amsthm}
\usepackage{amssymb}
\usepackage{geometry}
\usepackage{microtype}
\usepackage[unicode=true,
 bookmarks=false,
 breaklinks=false,pdfborder={0 0 1},backref=false,colorlinks=true]
 {hyperref}
\hypersetup{
 citecolor=blue!55!black,linkcolor=blue!55!black,urlcolor=blue!55!black}

\makeatletter
\numberwithin{equation}{section}
\numberwithin{figure}{section}

\usepackage{xcolor}
\makeatother

\theoremstyle{plain}
\newtheorem{thm}{\protect\theoremname}
\theoremstyle{definition}
\newtheorem{defn}[thm]{\protect\definitionname}
\theoremstyle{plain}
\newtheorem{prop}[thm]{\protect\propositionname}
\theoremstyle{remark}
\newtheorem{rem}[thm]{\protect\remarkname}
\theoremstyle{plain}
\newtheorem{lem}[thm]{\protect\lemmaname}
\newtheorem{cor}[thm]{\protect\corollaryname}
\providecommand{\corollaryname}{Corollary}
\providecommand{\definitionname}{Definition}
\providecommand{\lemmaname}{Lemma}
\providecommand{\propositionname}{Proposition}
\providecommand{\remarkname}{Remark}
\providecommand{\theoremname}{Theorem}

\begin{document}
\title{Non-Uniform Random Scans in Gibbs Sampling and CAVI}
\author{Sam Power}
\begin{abstract}
Gibbs sampling and coordinate ascent variational inference (CAVI)
are two basic coordinate-wise methods for statistical computation.
Recent analyses under strong log-concavity establish convergence rates
for versions of these algorithms that update one uniformly selected
block at each step. We extend both results to arbitrary fixed, strictly
positive selection probabilities. The rates are governed by a selection-adapted
convexity constant $\lambda^{\star}_{\theta}$, defined using the
block-smoothness constants and the selection probabilities $\theta$.
The same constant yields a contraction of relative entropy for the
Gibbs sampler and a contraction of the mean-field objective gap for
random-scan CAVI. The new bounds recover the uniform-scan results,
and are never weaker than the naive comparison based on the smallest
selection probability. They provide a principled way to adapt the
scan to heterogeneous block geometry using curvature information. 
\end{abstract}

\maketitle

\section{Introduction}

Many statistical algorithms update a high-dimensional parameter one
block at a time. Two canonical examples are the Gibbs sampler, which
replaces a block by a draw from its full conditional distribution,
and coordinate ascent variational inference (CAVI), which replaces
one factor of a mean-field approximation by its exact coordinate minimiser.
Random-scan versions select the block to update at random. They are
simple to implement, preserve the coordinate structure of the model,
and admit close analogies with randomised coordinate descent methods
in optimisation.

Recent work gives sharp convergence guarantees for these methods under
strong log-concavity. Ascolani, Lavenant, and Zanella prove contraction
of relative entropy for the random-scan Gibbs sampler \cite{ascolani2026entropy},
while Lavenant and Zanella prove geometric convergence of random-scan
CAVI \cite{lavenant2024cavi}. Both analyses select the blocks uniformly.
The purpose of this paper is to extend their conclusions to an arbitrary
fixed scan distribution 
\[
\theta=\left(\theta_{1},\ldots,\theta_{M}\right),\qquad\theta_{m}>0,\qquad\sum^{M}_{m=1}\theta_{m}=1.
\]
Although the change in the algorithm is elementary, the uniform-scan
proofs do not extend by a direct substitution of weights. Our main
contribution is to identify the geometry in which the selection probabilities
and the curvature of the target combine correctly. 

There are good practical reasons to consider non-uniform scans. Early
work of Levine and Casella develops random-scan Gibbs samplers that
learn or optimise their selection probabilities \cite{levine2006optimizing}.
Chimisov, Łatuszyński, and Roberts propose an adaptive random-scan
Gibbs sampler guided by the spectral gap of a Gaussian surrogate target,
and report practical gains for truncated Gaussian distributions, hierarchical
models, and hidden Markov models \cite{chimisov2018adapting}. Related
application-specific schemes select latent coordinates or observations
more frequently when their updates are expected to be informative
\cite{casarin2024multiple,fabbrico2025discomfort}. Adaptive component-wise
methods also have a longer history in Bayesian variable selection
\cite{nott2005adaptive}. These works address different adaptation
criteria, but collectively show that uniform allocation can waste
substantial computation in heterogeneous models, and that non-uniform
strategies can be of substantial practical interest.

The quantity governing our results is the \emph{selection-adapted
convexity constant} $\lambda^{\star}_{\theta}$. For a target distribution
$\pi\propto\exp\left(-U\right)$, this measures the convexity of the
non-separable part of the potential $U$ in a norm which is informed
by both the block-smoothness structure of $U$ and the frequencies
with which each block is updated. Our first main result proves that
the non-uniform Gibbs kernel contracts relative entropy by the factor
$1-\lambda^{\star}_{\theta}$. Our second proves that random-scan
CAVI contracts the expected suboptimality of the mean-field variational
inference objective by the same factor. With uniform probabilities,
both results reduce to the rates in \cite{ascolani2026entropy,lavenant2024cavi}.
For a general scan, they are always at least as strong as the crude
bound obtained by comparing every $\theta_{m}$ with $\min_{j}\theta_{j}$,
and in specific cases, they can be considerably better.

The paper is organised as follows. Section~\ref{sec:setup} introduces
the two algorithms and selection-adapted convexity. Section~\ref{sec:results}
states the main convergence results and compares them with uniform
and naive bounds. Section~\ref{sec:examples} develops tools for
Gaussian scan design and illustrates the possible gains in statistical
models. Section~\ref{subsec:interpolation} isolates the common weighted
interpolation estimate. Sections~\ref{subsec:gibbs-proof} and \ref{subsec:cavi-proof}
give the Gibbs and CAVI proofs, respectively. Section~\ref{sec:discussion}
discusses scan design, reparameterisation, and computational costs.

\section{Coordinate-wise computation and selection-adapted convexity}\label{sec:setup}

\subsection{Target distribution and block geometry}

Let $\mathcal{X}=\mathcal{X}_{1}\times\cdots\times\mathcal{X}_{M}$,
$\mathcal{X}_{m}=\mathbb{R}^{d_{m}}$, $d=\sum^{M}_{m=1}d_{m}$. For
$x=\left(x_{1},\ldots,x_{M}\right)$, write $x_{-m}$ for all blocks
except $x_{m}$, and write $\left(y_{m},x_{-m}\right)$ for the point
obtained by replacing $x_{m}$ by $y_{m}$. Let $\pi$ have normalised
Lebesgue density 
\begin{equation}
\pi\left(\mathrm{d}x\right)=\exp\left(-U\left(x\right)\right)\,\mathrm{d}x,\qquad\int_{\mathbb{R}^{d}}\exp\left(-U\left(x\right)\right)\,\mathrm{d}x=1.\label{eq:target}
\end{equation}

We use the following blockwise version of the assumptions in \cite{ascolani2026entropy}.

\medskip{}
\noindent \textbf{Assumption A.} The potential decomposes as 
\begin{equation}
U\left(x\right)=U_{0}\left(x\right)+\sum^{M}_{m=1}U_{m}\left(x_{m}\right),\label{eq:decomposition}
\end{equation}
where each $U_{m}:\mathbb{R}^{d_{m}}\to\mathbb{R}$ is convex. The
function $U_{0}\in C^{1}\left(\mathbb{R}^{d}\right)$ is convex and
block-$L_{m}$-smooth: for every $x\in\mathbb{R}^{d}$, $y_{m}\in\mathbb{R}^{d_{m}}$,
and $m$, 
\begin{equation}
\begin{aligned}U_{0}\left(y_{m},x_{-m}\right) & \leq U_{0}\left(x\right)+\left\langle \nabla_{m}U_{0}\left(x\right),y_{m}-x_{m}\right\rangle +\frac{L_{m}}{2}\left\Vert y_{m}-x_{m}\right\Vert ^{2},\end{aligned}
\label{eq:block-smooth}
\end{equation}
where $L_{m}>0$.

Fix a scan distribution $\theta$ in the interior of the probability
simplex and define 
\begin{equation}
\left\Vert x\right\Vert ^{2}_{L,\theta}:=\sum^{M}_{m=1}\frac{L_{m}}{\theta_{m}}\left\Vert x_{m}\right\Vert ^{2},\qquad\mathbf{D}_{\theta}:=\mathsf{diag}\left(\frac{L_{1}}{\theta_{1}}\mathbf{I}_{d_{1}},\ldots,\frac{L_{M}}{\theta_{M}}\mathbf{I}_{d_{M}}\right),\qquad\mathbf{G}_{\theta}:=\mathbf{D}^{-1}_{\theta}.\label{eq:weighted-geometry}
\end{equation}

\begin{defn}[Selection-adapted convexity]
The selection-adapted convexity constant $\lambda^{\star}_{\theta}$
is the largest $\lambda\geq0$ such that 
\begin{equation}
U_{0}\left(y\right)\geq U_{0}\left(x\right)+\left\langle \nabla U_{0}\left(x\right),y-x\right\rangle +\frac{\lambda}{2}\left\Vert y-x\right\Vert ^{2}_{L,\theta}\label{eq:selection-convexity}
\end{equation}
for every $x,y\in\mathbb{R}^{d}$. If $U_{0}\in C^{2}\left(\mathbb{R}^{d}\right)$,
then this is equivalent to $\nabla^{2}U_{0}\left(x\right)\succeq\lambda^{\star}_{\theta}\mathbf{D}_{\theta}$
for every $x\in\mathbb{R}^{d}$.
\end{defn}

Block smoothness implies that $0\leq\lambda^{\star}_{\theta}\leq\theta_{\min}:=\min_{1\leq m\leq M}\theta_{m}$.
Indeed, restricting \eqref{eq:selection-convexity} to block $m$
shows that $\lambda^{\star}_{\theta}L_{m}/\theta_{m}\leq L_{m}$.

The decomposition \eqref{eq:decomposition} is not canonical: separable
terms may be allocated between $U_{0}$ and the functions $U_{m}$
in more than one way, provided that Assumption~A remains satisfied.
The value of $\lambda^{\star}_{\theta}$ should therefore be understood
relative to the chosen decomposition, the chosen valid smoothness
constants, and the Euclidean coordinates within each block. 

The factors $L_{m}/\theta_{m}$ give the weighted geometry a direct
interpretation. Reducing $\theta_{m}$ enlarges the metric in block
$m$ and makes the convexity inequality harder to satisfy in that
direction. $\lambda^{\star}_{\theta}$ thereby penalises scans which
update a block too infrequently relative to its smoothness and its
role in weakly curved, coupled directions. Conversely, probability
can be shifted away from directions whose curvature is ample relative
to their block smoothness and towards those which limit the contraction.
The bound $\lambda^{\star}_{\theta}\leq\theta_{\min}$ records that
no block can be ignored; improvements over uniform scanning arise
by matching the allocation to the anisotropy of the target.

\subsection{The non-uniform Gibbs sampler}

For each $m$, let $P_{m}$ be the Gibbs kernel that leaves $x_{-m}$
fixed and samples the new block from $\pi\left(\mathrm{d}x_{m}\mid x_{-m}\right)$.
The non-uniform random-scan kernel is 
\begin{equation}
P^{\theta}:=\sum^{M}_{m=1}\theta_{m}P_{m}.\label{eq:gibbs-kernel}
\end{equation}
The chain rule for relative entropy and exact conditional resampling
give the variational identity 
\begin{equation}
\mathsf{KL}\left(\mu P_{m}\mid\pi\right)=\mathsf{KL}\left(\mu_{-m}\mid\pi_{-m}\right)=\inf_{\nu:\,\nu_{-m}=\mu_{-m}}\mathsf{KL}\left(\nu\mid\pi\right),\label{eq:gibbs-variational}
\end{equation}
whenever $\mathsf{KL}\left(\mu\mid\pi\right)<+\infty$.

\subsection{Non-uniform random-scan CAVI}

Let $\mathcal{P}^{\otimes M}_{2}\left(\mathcal{X}\right)$ denote
the product probability measures on $\mathcal{X}$ with finite second
moments, and define 
\begin{equation}
F\left(q\right):=\mathsf{KL}\left(q\mid\pi\right).\label{eq:mean-field-objective}
\end{equation}
Whenever $\lambda^{\star}_{\theta}>0$, this objective has a unique
minimiser over $\mathcal{P}^{\otimes M}_{2}\left(\mathcal{X}\right)$,
which we denote by 
\begin{equation}
q^{\star}:=\arg\min_{q\in\mathcal{P}^{\otimes M}_{2}\left(\mathcal{X}\right)}F\left(q\right).\label{eq:mean-field-minimizer}
\end{equation}
Indeed, strong convexity of $U_{0}$ gives coercivity, so existence
follows from the direct method and lower semicontinuity of relative
entropy; the class of product measures is closed under weak limits.
Uniqueness follows from strict displacement convexity along factorwise
quadratic Wasserstein geodesics. The same argument gives existence
and uniqueness of each coordinate minimiser below whenever $F\left(q\right)<+\infty$. 

For $q=q_{1}\otimes\cdots\otimes q_{M}$, the exact CAVI update of
factor $m$ is 
\begin{equation}
\mathcal{C}_{m}q:=\arg\min_{\nu_{m}\in\mathcal{P}_{2}\left(\mathcal{X}_{m}\right)}F\left(\nu_{m}\otimes q_{-m}\right).\label{eq:cavi-update}
\end{equation}
Equivalently, its updated factor has density 
\begin{equation}
\left(\mathcal{C}_{m}q\right)_{m}\left(\mathrm{d}x_{m}\right)=\frac{1}{Z_{m}\left(q_{-m}\right)}\exp\left(-\int_{\mathcal{X}_{-m}}U\left(x_{m},x_{-m}\right)q_{-m}\left(\mathrm{d}x_{-m}\right)\right)\,\mathrm{d}x_{m}.\label{eq:cavi-density}
\end{equation}
Let $I_{1},I_{2},\ldots$ be independent and identically distributed
with $\mathbf{P}\left(I_{n}=m\right)=\theta_{m}$, and set $q^{\left(n\right)}=\mathcal{C}_{I_{n}}q^{\left(n-1\right)}$.
These iterations define the random-scan CAVI algorithm.

\section{Main results}\label{sec:results}

Our first result is a weighted analogue of the approximate tensorisation
and entropy contraction results of \cite{ascolani2026entropy}.
\begin{thm}[Weighted entropy contraction]
\label{thm:gibbs} Suppose that Assumption~A holds and $\lambda^{\star}_{\theta}>0$.
For every probability measure $\mu$ satisfying $\mathsf{KL}\left(\mu\mid\pi\right)<+\infty$,
\begin{equation}
\sum^{M}_{m=1}\theta_{m}\mathsf{KL}\left(\mu_{-m}\mid\pi_{-m}\right)\leq\left(1-\lambda^{\star}_{\theta}\right)\mathsf{KL}\left(\mu\mid\pi\right).\label{eq:weighted-tensorization}
\end{equation}
Equivalently, 
\begin{equation}
\lambda^{\star}_{\theta}\mathsf{KL}\left(\mu\mid\pi\right)\leq\sum^{M}_{m=1}\theta_{m}\mathbf{E}_{\mu_{-m}}\left[\mathsf{KL}\left(\mu\left(\mathord{\cdot}\mid X_{-m}\right)\mathrel{\big|}\pi\left(\mathord{\cdot}\mid X_{-m}\right)\right)\right].\label{eq:weighted-conditional}
\end{equation}
Consequently, 
\begin{equation}
\mathsf{KL}\left(\mu P^{\theta}\mid\pi\right)\leq\left(1-\lambda^{\star}_{\theta}\right)\mathsf{KL}\left(\mu\mid\pi\right).\label{eq:gibbs-contraction}
\end{equation}
\end{thm}

Iterating \eqref{eq:gibbs-contraction} gives 
\begin{equation}
\mathsf{KL}\left(\mu\left(P^{\theta}\right)^{n}\mid\pi\right)\leq\left(1-\lambda^{\star}_{\theta}\right)^{n}\mathsf{KL}\left(\mu\mid\pi\right).\label{eq:gibbs-iterate}
\end{equation}

The corresponding CAVI result contracts at precisely the same rate.
\begin{thm}[Weighted convergence of random-scan CAVI]
\label{thm:cavi} Suppose that Assumption~A holds and $\lambda^{\star}_{\theta}>0$.
If $q\in\mathcal{P}^{\otimes M}_{2}\left(\mathcal{X}\right)$ and
$F\left(q\right)<+\infty$, then 
\begin{equation}
\sum^{M}_{m=1}\theta_{m}\left[F\left(\mathcal{C}_{m}q\right)-F\left(q^{\star}\right)\right]\leq\left(1-\lambda^{\star}_{\theta}\right)\left[F\left(q\right)-F\left(q^{\star}\right)\right].\label{eq:cavi-one-step}
\end{equation}
Consequently, if $F\left(q^{\left(0\right)}\right)<+\infty$, then
\begin{equation}
\begin{aligned} & \mathbf{E}\left[F\left(q^{\left(n+1\right)}\right)-F\left(q^{\star}\right)\,\middle|\,q^{\left(n\right)}\right]\\
 & \hspace{3cm}\leq\left(1-\lambda^{\star}_{\theta}\right)\left[F\left(q^{\left(n\right)}\right)-F\left(q^{\star}\right)\right],
\end{aligned}
\label{eq:cavi-conditional}
\end{equation}
and 
\begin{equation}
\mathbf{E}\left[F\left(q^{\left(n\right)}\right)\right]-F\left(q^{\star}\right)\leq\left(1-\lambda^{\star}_{\theta}\right)^{n}\left[F\left(q^{\left(0\right)}\right)-F\left(q^{\star}\right)\right].\label{eq:cavi-iterate}
\end{equation}
\end{thm}

\subsection{Uniform scans and the naive comparison}

Let $\mathbf{D}:=\mathsf{diag}\left(L_{1}\mathbf{I}_{d_{1}},\ldots,L_{M}\mathbf{I}_{d_{M}}\right)$,
and let $\lambda^{\star}$ be the largest constant such that $\nabla^{2}U_{0}\left(x\right)\succeq\lambda^{\star}\mathbf{D}$
for every $x$, with the equivalent first-order definition when $U_{0}$
is only $C^{1}$.
\begin{prop}[Recovery and dominance]
\label{prop:comparison} For the uniform scan $\theta_{m}=1/M$,
it holds that $\lambda^{\star}_{\mathrm{unif}}=\frac{\lambda^{\star}}{M}$.
For every scan distribution $\theta$, it holds that $\lambda^{\star}_{\theta}\geq\theta_{\min}\lambda^{\star}=M\theta_{\min}\lambda^{\star}_{\mathrm{unif}}$.
\end{prop}

\begin{rem}
Theorems~\ref{thm:gibbs} and \ref{thm:cavi} are thus never weaker
than the rate obtained by combining the corresponding uniform-scan
theorem with the elementary comparison based on $\theta_{\min}$.
\end{rem}

\begin{proof}
For the uniform scan, $\mathbf{D}_{\mathrm{unif}}=M\mathbf{D}$, proving
the first claim. Moreover, $\mathbf{D}_{\theta}\preceq\frac{1}{\theta_{\min}}\mathbf{D}$,
and hence $\nabla^{2}U_{0}\succeq\lambda^{\star}\mathbf{D}$ implies
$\nabla^{2}U_{0}\succeq\theta_{\min}\lambda^{\star}\mathbf{D}_{\theta}$,
which proves the second claim.
\end{proof}

For comparison, the elementary comparison for Gibbs writes $P^{\theta}$
as a mixture containing the uniform kernel with weight $M\theta_{\min}$,
uses convexity of relative entropy together with the fact that every
$P_{m}$ is entropy non-increasing, and gives the contraction parameter
$\theta_{\min}\lambda^{\star}$. For CAVI, decomposing the weighted
average of $F\left(\mathcal{C}_{m}q\right)$ in the same way and using
$F\left(\mathcal{C}_{m}q\right)\leq F\left(q\right)$ again gives
$\theta_{\min}\lambda^{\star}$, which is no larger than $\lambda^{\star}_{\theta}$. 

\section{Scan design and statistical examples}\label{sec:examples}

This section illustrates statistical settings in which non-uniform
scanning gives an appreciable improvement in the theoretical convergence
rates above. The point is to isolate mechanisms for this improvement
that can be recognised (or not) in other problems. We first record
two general results which will aid in our study of Gaussian targets,
and then use them to study some linear-Gaussian hierarchical models
in detail. A separate smoothness-based example shows that substantial
gains can also be certified without knowing the full Hessian. Throughout,
an \emph{improvement} will denote the ratio of the contraction parameters
$\lambda^{\star}_{\theta}/\lambda^{\star}_{\mathrm{unif}}$. When
these parameters are small, this agrees to first order with the ratio
of the corresponding exponential rates $-\log\left(1-\lambda^{\star}_{\theta}\right)$
and $-\log\left(1-\lambda^{\star}_{\mathrm{unif}}\right)$.

\subsection{Gaussian scan optimisation}\label{subsec:gaussian-design}

Consider the Gaussian potential $U\left(x\right)=\frac{1}{2}x^{\mathsf{T}}\mathbf{Q}x-\ell^{\mathsf{T}}x+c$,
where $\mathbf{Q}$ is symmetric positive definite. Write $\mathbf{D}_{Q}:=\mathsf{diag}\left(\mathbf{Q}_{11},\ldots,\mathbf{Q}_{MM}\right)$
and $\mathbf{R}:=\mathbf{D}^{-1/2}_{Q}\mathbf{Q}\mathbf{D}^{-1/2}_{Q}$,
so that $\mathbf{R}$ is the block-normalised precision matrix, with
identity diagonal blocks.
\begin{prop}[Gaussian scan design after block whitening]
\label{prop:gaussian-design} Take $U_{0}=U$ and $U_{m}=0$ for
every $m$, and express the potential in the block-whitened coordinates
$z=\mathbf{D}^{1/2}_{Q}x$. In these coordinates, put $\mathbf{G}_{\theta}=\mathsf{diag}\left(\theta_{1}\mathbf{I}_{d_{1}},\ldots,\theta_{M}\mathbf{I}_{d_{M}}\right)$.
Then
\begin{equation}
\lambda^{\star}_{\theta}=\lambda_{\min}\left(\mathbf{G}^{1/2}_{\theta}\mathbf{R}\mathbf{G}^{1/2}_{\theta}\right).\label{eq:gaussian-design-lambda}
\end{equation}
Moreover, maximising $\lambda^{\star}_{\theta}$ over scan distributions
is equivalent to the convex optimisation problem
\begin{equation}
\left(\lambda^{\star}_{\mathrm{opt}}\right)^{-1}=\min_{a_{1},\ldots,a_{M}>0}\left\{ \sum^{M}_{m=1}a^{-1}_{m}:\ \mathbf{R}\succeq\mathsf{diag}\left(a_{1}\mathbf{I}_{d_{1}},\ldots,a_{M}\mathbf{I}_{d_{M}}\right)\right\} .\label{eq:gaussian-design-convex}
\end{equation}
If $a^{\star}$ solves \eqref{eq:gaussian-design-convex}, then an
optimal scan is given by $\theta^{\star}_{m}=\lambda^{\star}_{\mathrm{opt}}/a^{\star}_{m}$.
\end{prop}

\begin{proof}
Both Gibbs sampling and the mean-field variational inference problem
are invariant under invertible changes of coordinates within each
block. In the $z$-coordinates, the Hessian is $\mathbf{R}$ and each
block smoothness constant may be taken to be one. Hence $\mathbf{D}_{\theta}=\mathbf{G}^{-1}_{\theta}$,
and the largest $\lambda$ for which $\mathbf{R}\succeq\lambda\mathbf{G}^{-1}_{\theta}$
is the smallest eigenvalue in \eqref{eq:gaussian-design-lambda}.
For a feasible pair $\left(\lambda,\theta\right)$, set $a_{m}=\lambda/\theta_{m}$.
The matrix inequality becomes the constraint in \eqref{eq:gaussian-design-convex},
while $\sum_{m}\theta_{m}=1$ gives $\lambda^{-1}=\sum_{m}a^{-1}_{m}$.
Conversely, any feasible $a$ defines a scan by $\theta_{m}=\lambda/a_{m}$
with this value of $\lambda$; the equivalence follows.
\end{proof}

While this convex optimisation problem need not admit a closed-form
solution, it makes the design problem explicit and relatively computationally
tractable. It also gives a useful interpretation of the optimal probabilities.
If the smallest eigenvalue in \eqref{eq:gaussian-design-lambda} is
simple at an interior optimum, with unit eigenvector $v=\left(v_{1},\ldots,v_{M}\right)$,
then first-order optimality gives $\theta^{\star}_{m}=\left\Vert v_{m}\right\Vert ^{2}$.
The scan hence allocates probability according to the blockwise mass
of its own slowest mode. The relation is implicit because that mode
itself depends on $\theta$; when the bottom eigenvalue is multiple,
the optimum instead balances the active slow directions.

We now include a short lemma which simplifies the treatment of jointly
Gaussian models for which the blocks enjoy certain structural symmetries,
as common for hierarchical models in Bayesian statistics.
\begin{lem}[Symmetrisation of the scan]
\label{lem:orbit-averaging} Let a finite group $\Gamma$ permute
equal-dimensional blocks, and let $\mathbf{P}_{g}$ be the orthogonal
block-permutation matrix induced by $g\in\Gamma$. Suppose that $\mathbf{P}^{\mathsf{T}}_{g}\mathbf{R}\mathbf{P}_{g}=\mathbf{R}$
for every $g\in\Gamma$. If $\bar{\theta}$ is obtained by averaging
a scan $\theta$ over $\Gamma$, then $\lambda^{\star}_{\bar{\theta}}\geq\lambda^{\star}_{\theta}$.
Consequently, an optimal scan may be taken to be constant on each
orbit of the group action.
\end{lem}

\begin{proof}
If $\mathbf{R}\succeq\lambda\mathbf{G}^{-1}_{\theta}$, then conjugating
by each $\mathbf{P}_{g}$ and averaging over the group gives that
\[
\mathbf{R}\succeq\frac{\lambda}{\lvert\Gamma\rvert}\sum_{g\in\Gamma}\mathbf{G}^{-1}_{\theta^{g}},
\]
where $\theta^{g}$ denotes the permuted scan. The convexity of $t\mapsto t^{-1}$
yields, block by block,
\[
\frac{1}{\lvert\Gamma\rvert}\sum_{g\in\Gamma}\mathbf{G}^{-1}_{\theta^{g}}\succeq\mathbf{G}^{-1}_{\bar{\theta}},
\]
and so $\lambda$ is also feasible for $\bar{\theta}$. The group
average is constant on every orbit by construction.
\end{proof}

\subsection{Localised difficulty}\label{subsec:localized}

The cleanest mechanism for a large gain is that the `slow' geometry
of the target distribution is confined to a small collection of blocks.
The following proposition isolates this mechanism.
\begin{prop}[A small hard subsystem]
\label{prop:localized-hardness} Suppose that the normalised Gaussian
precision decomposes as $\mathbf{R}=\mathbf{R}_{H}\oplus\mathbf{I}$,
where $\mathbf{R}_{H}$ contains $k$ blocks, $\lambda_{\min}\left(\mathbf{R}_{H}\right)=\varepsilon\in\left(0,1\right]$,
and $\mathbf{I}$ contains $M-k$ blocks which are each independent
and normalised. Assign probability
\[
a:=\frac{1}{k+\left(M-k\right)\varepsilon}
\]
to each `hard' block and probability $b:=\varepsilon a$ to each `easy'
block. Then
\begin{equation}
\lambda^{\star}_{\theta}=\frac{\varepsilon}{k+\left(M-k\right)\varepsilon}.\label{eq:localized-lambda}
\end{equation}
For the uniform scan, $\lambda^{\star}_{\mathrm{unif}}=\varepsilon/M$,
and hence
\begin{equation}
\frac{\lambda^{\star}_{\theta}}{\lambda^{\star}_{\mathrm{unif}}}=\frac{M}{k+\left(M-k\right)\varepsilon}\longrightarrow\frac{M}{k}\qquad\text{as }\varepsilon\downarrow0.\label{eq:localized-gain}
\end{equation}
In particular, along regimes in which $\varepsilon\downarrow0$ and
$k/M\to0$, the improvement over the uniform scan is unbounded.
\end{prop}

\begin{proof}
For the displayed scan, $\mathbf{G}^{1/2}_{\theta}\mathbf{R}\mathbf{G}^{1/2}_{\theta}=a\mathbf{R}_{H}\oplus b\mathbf{I}$.
Proposition~\ref{prop:gaussian-design} therefore gives $\lambda^{\star}_{\theta}=\min\left\{ a\varepsilon,b\right\} $.
Since $a\varepsilon=b$ and $ka+\left(M-k\right)b=1$, this proves
\eqref{eq:localized-lambda}. Under the uniform scan, the weighted
normalised precision is $\mathbf{R}/M$, from which the remaining
claims follow.
\end{proof}

The gain is large because uniform scanning spends almost all of its
updates on blocks that do not participate in the nearly-singular direction
in the target, wasting effort on blocks which are already well-resolved
by the algorithm. The proposed scan instead balances the rate $a\varepsilon$
of the hard subsystem against the rate $b$ of an easy block. 

This geometry arises naturally in hierarchical regression when only
a few groups are strongly confounded with the global effects. Consider
the linear random effects model
\begin{equation}
y_{i}=\mathbf{X}_{i}\beta+\mathbf{Z}_{i}b_{i}+\varepsilon_{i},\qquad i=1,\ldots,G,\label{eq:random-slopes-model-example}
\end{equation}
with independent Gaussian priors $\beta\sim\mathcal{N}\left(0,\boldsymbol{\Sigma}_{\beta}\right)$
and $b_{i}\sim\mathcal{N}\left(0,\boldsymbol{\Sigma}_{b}\right)$,
and with errors having precision matrices $\mathbf{W}_{i}$. The posterior
precision blocks are
\begin{equation}
\begin{aligned}\mathbf{Q}_{\beta\beta} & =\boldsymbol{\Sigma}^{-1}_{\beta}+\sum^{G}_{i=1}\mathbf{X}^{\mathsf{T}}_{i}\mathbf{W}_{i}\mathbf{X}_{i},\\
\mathbf{Q}_{b_{i}b_{i}} & =\boldsymbol{\Sigma}^{-1}_{b}+\mathbf{Z}^{\mathsf{T}}_{i}\mathbf{W}_{i}\mathbf{Z}_{i},\\
\mathbf{Q}_{\beta b_{i}} & =\mathbf{X}^{\mathsf{T}}_{i}\mathbf{W}_{i}\mathbf{Z}_{i},\qquad\mathbf{Q}_{b_{i}b_{j}}=\mathbf{0}\quad\text{for }i\neq j.
\end{aligned}
\label{eq:random-slopes-precision-example}
\end{equation}
Define the normalised global-{}-local couplings $\mathbf{C}_{i}:=\mathbf{Q}^{-1/2}_{\beta\beta}\mathbf{Q}_{\beta b_{i}}\mathbf{Q}^{-1/2}_{b_{i}b_{i}}$.
\begin{cor}[A few confounded clusters]
\label{cor:few-confounded-clusters} Suppose, after reordering, that
$\mathbf{C}_{i}=\mathbf{0}$ for $k<i\leq G$. Let $\mathbf{R}_{H}$
be the normalised precision restricted to the global block and the
first $k$ random-effect blocks, and suppose that $\lambda_{\min}\left(\mathbf{R}_{H}\right)=\varepsilon\in\left(0,1\right]$.
Assign probability $a=\left[k+1+\left(G-k\right)\varepsilon\right]^{-1}$
to the global block and each of the first $k$ group blocks, and probability
$b=\varepsilon a$ to every remaining group block. Then
\[
\lambda^{\star}_{\theta}=\frac{\varepsilon}{k+1+\left(G-k\right)\varepsilon},\qquad\lambda^{\star}_{\mathrm{unif}}=\frac{\varepsilon}{G+1},
\]
and
\begin{equation}
\frac{\lambda^{\star}_{\theta}}{\lambda^{\star}_{\mathrm{unif}}}=\frac{G+1}{k+1+\left(G-k\right)\varepsilon}\longrightarrow\frac{G+1}{k+1}\qquad\text{as }\varepsilon\downarrow0.\label{eq:few-confounded-gain}
\end{equation}
\end{cor}

\begin{proof}
Whitening each block by its diagonal precision gives identity diagonal
blocks and off-diagonal blocks $\mathbf{C}_{i}$ between the global
block and group $i$. There are no interactions between distinct group
blocks, whereby the stated condition gives $\mathbf{R}=\mathbf{R}_{H}\oplus\mathbf{I}$,
and Proposition~\ref{prop:localized-hardness} applies with $M=G+1$
and $k+1$ hard blocks.
\end{proof}

The condition $\mathbf{C}_{i}=\mathbf{0}$ says that the fixed- and
random-effect designs are orthogonal in the $\mathbf{W}_{i}$-geometry
for the easy groups. Small rather than zero values of $\left\Vert \mathbf{C}_{i}\right\Vert _{\mathrm{op}}$
perturb the displayed rates continuously, while singular values near
one identify groups with nearly confounded fixed and random effects.
The example hence describes a regime in which the number of \emph{statistically
difficult} groups, rather than the total number of groups, should
determine how the scan allocates most of its probability.

\subsection{Heterogeneous block smoothness}\label{subsec:smoothness-design}

The previous arguments exploit the full Gaussian precision. A complementary
mechanism is visible from examination of only the block smoothness
constants, and applies beyond the Gaussian setting.
\begin{prop}[Importance probabilities from smoothness]
\label{prop:smoothness-design} Suppose that $\nabla^{2}U_{0}\left(x\right)\succeq\mu\mathbf{I}_{d}$
for every $x$, and that $U_{0}$ is block-$L_{m}$-smooth. Define
$\theta^{L}_{m}:=L_{m}/\sum^{M}_{j=1}L_{j}$. Then
\begin{equation}
\lambda^{\star}_{\theta^{L}}\geq\frac{\mu}{\sum^{M}_{j=1}L_{j}}.\label{eq:smoothness-design}
\end{equation}
For the uniform scan, the same coarse information gives
\begin{equation}
\lambda^{\star}_{\mathrm{unif}}\geq\frac{\mu}{ML_{\max}},\qquad L_{\max}:=\max_{1\leq m\leq M}L_{m}.\label{eq:smoothness-uniform-certificate}
\end{equation}
Consequently, the ratio between these two computable lower bounds
is
\begin{equation}
\frac{ML_{\max}}{\sum^{M}_{m=1}L_{m}},\label{eq:smoothness-certificate-gain}
\end{equation}
which lies between $1$ and $M$, and can be arbitrarily close to
$M$.
\end{prop}

\begin{proof}
For $\theta^{L}$, $\mathbf{D}_{\theta^{L}}=\left(\sum^{M}_{j=1}L_{j}\right)\mathbf{I}_{d}$,
which proves \eqref{eq:smoothness-design}. For the uniform scan,
\[
\mathbf{D}_{\mathrm{unif}}=M\mathsf{diag}\left(L_{1}\mathbf{I}_{d_{1}},\ldots,L_{M}\mathbf{I}_{d_{M}}\right)\preceq ML_{\max}\mathbf{I}_{d},
\]
which proves \eqref{eq:smoothness-uniform-certificate} and the stated
ratio.
\end{proof}

The statistical relevance of this result is that block scales commonly
differ because of unequal group sizes, noise levels, or amounts of
design information. When detailed dependence information is unavailable,
the smoothness rule is an explicit and robust way to prevent the largest-scale
block from causing the convergence rate estimate to deteriorate too
aggressively. Similar considerations are discussed in \cite{nesterov2012coordinate}.

\subsection{Balanced nested random intercepts}\label{subsec:nested-example}

The preceding random-slopes example localises the difficult direction
in a few groups. A balanced random-intercepts model exhibits a contrasting
case in which the difficulty is spread symmetrically across all groups.
Consider
\begin{equation}
y_{ij}=\beta+u_{i}+\varepsilon_{ij},\qquad i=1,\ldots,G,\quad j=1,\ldots,n,\label{eq:nested-model-example}
\end{equation}
with independent priors and errors
\[
\beta\sim\mathcal{N}\left(0,\tau^{-1}_{\beta}\right),\qquad u_{i}\sim\mathcal{N}\left(0,\tau^{-1}_{u}\right),\qquad\varepsilon_{ij}\sim\mathcal{N}\left(0,\tau^{-1}_{e}\right).
\]
Take $\beta,u_{1},\ldots,u_{G}$ as the $G+1$ blocks and define
\[
A:=\tau_{\beta}+\tau_{e}Gn,\qquad B:=\tau_{u}+\tau_{e}n,\qquad r:=\frac{\tau_{e}n}{\sqrt{AB}},\qquad c:=\sqrt{G}\,r.
\]
Positive definiteness gives $0<c<1$.
\begin{thm}[Balanced nested random intercepts]
\label{thm:nested-intercepts} An optimal scan may be taken to have
$\theta_{\beta}=s$ and $\theta_{u_{i}}=b:=\left(1-s\right)/G$ for
$1\leq i\leq G$. For this scan,
\begin{equation}
\Lambda_{G}\left(s\right):=\lambda^{\star}_{\theta}=\frac{1}{2}\left[s+b-\sqrt{\left(s-b\right)^{2}+4r^{2}s\left(1-s\right)}\right].\label{eq:nested-example-lambda}
\end{equation}
Under the uniform scan,
\begin{equation}
\lambda^{\star}_{\mathrm{unif}}=\frac{1-c}{G+1}.\label{eq:nested-example-uniform}
\end{equation}
Let $s_{c}$ maximise $\Lambda_{G}$. In the weak-prior, nearly confounded
regime $c\uparrow1$,
\begin{equation}
s_{c}\longrightarrow\frac{1}{1+\sqrt{G}},\qquad\frac{\theta_{\beta}}{\theta_{u_{i}}}\longrightarrow\sqrt{G},\qquad\lambda^{\star}_{\mathrm{opt}}\sim\frac{1-c^{2}}{\left(1+\sqrt{G}\right)^{2}}.\label{eq:nested-example-optimum}
\end{equation}
Consequently,
\begin{equation}
\frac{\lambda^{\star}_{\mathrm{opt}}}{\lambda^{\star}_{\mathrm{unif}}}\longrightarrow\frac{2\left(G+1\right)}{\left(1+\sqrt{G}\right)^{2}},\label{eq:nested-example-gain}
\end{equation}
and the last expression tends to $2$ as $G\to\infty$.
\end{thm}

\begin{proof}
The block-normalised posterior precision is
\[
\mathbf{R}=\begin{pmatrix}1 & r\mathbf{1}^{\mathsf{T}}_{G}\\
r\mathbf{1}_{G} & \mathbf{I}_{G}
\end{pmatrix}.
\]
Lemma~\ref{lem:orbit-averaging} shows that an optimal scan may assign
the same probability $b$ to every $u_{i}$. Every contrast orthogonal
to $\mathbf{1}_{G}$ is then an eigenvector of the weighted normalised
precision with eigenvalue $b$. On the span of the global factor and
$G^{-1/2}\mathbf{1}_{G}$, the matrix is
\[
\mathbf{K}_{s}=\begin{pmatrix}s & r\sqrt{s\left(1-s\right)}\\
r\sqrt{s\left(1-s\right)} & b
\end{pmatrix}.
\]
Its smaller eigenvalue is below $b$, so Proposition~\ref{prop:gaussian-design}
gives \eqref{eq:nested-example-lambda}. At the uniform scan, the
two eigenvalues of $\mathbf{K}_{s}$ are $\left(1\pm c\right)/\left(G+1\right)$,
proving \eqref{eq:nested-example-uniform}. For the asymptotic optimisation,
set $t:=s/b=Gs/\left(1-s\right)$, so that $s=t/\left(G+t\right)$
and $b=1/\left(G+t\right)$. Then
\[
\Lambda_{G}\left(s\right)=\frac{1}{G+t}\lambda_{\min}\begin{pmatrix}t & c\sqrt{t}\\
c\sqrt{t} & 1
\end{pmatrix}.
\]
The determinant of the last matrix is $t\left(1-c^{2}\right)$, and
its larger eigenvalue tends to $t+1$ as $c\uparrow1$. Hence, for
fixed $t>0$,
\[
\frac{\Lambda_{G}\left(s\right)}{1-c^{2}}\longrightarrow\frac{t}{\left(t+1\right)\left(t+G\right)}.
\]
The limiting function has the unique maximiser $t=\sqrt{G}$; the
explicit eigenvalue formula excludes maximising sequences approaching
either endpoint. This proves \eqref{eq:nested-example-optimum}, and
division by \eqref{eq:nested-example-uniform} gives \eqref{eq:nested-example-gain}.
\end{proof}

Here the global mean is coupled diffusely to every group effect. The
optimal scan is nevertheless strongly non-uniform: it updates the
global block approximately $\sqrt{G}$ times as often as each local
block. Unlike the localised example, however, the gain in the contraction
parameter remains bounded and approaches two as the number of groups
grows.

These examples distinguish three sources of scan heterogeneity. A
difficult subsystem confined to $k$ of $M$ blocks can produce a
gain of order $M/k$; unequal block smoothness can yield an improvement
of order $M$ in guarantees based on coarse curvature information,
and diffuse hierarchical coupling can require markedly unequal probabilities
for optimality while giving only a bounded gain in overall convergence.
In each case the scan responds to the location of the slow geometry,
rather than merely to the number or labels of the blocks. The calculations
concern convergence per iteration; the separate effect of unequal
update costs is discussed in Section~\ref{sec:discussion}.

\section{Proofs}

\subsection{The common weighted interpolation estimate}\label{subsec:interpolation}

The two main proofs share the same deterministic estimate for the
potential. The algorithm-specific distinction appears only when the
entropy terms are handled.
\begin{lem}[Weighted coordinate interpolation]
\label{lem:potential} Suppose that Assumption~A holds and choose
$0<\lambda\leq\lambda^{\star}_{\theta}$. Set $t_{m}:=\frac{\lambda}{\theta_{m}}\in\left[0,1\right]$,
and, for $x,y\in\mathbb{R}^{d}$, define  $S_{m}\left(x,y\right):=\left(x_{-m},\left(1-t_{m}\right)x_{m}+t_{m}y_{m}\right)$.
Then 
\begin{equation}
\sum^{M}_{m=1}\theta_{m}U\left(S_{m}\left(x,y\right)\right)\leq\left(1-\lambda\right)U\left(x\right)+\lambda U\left(y\right).\label{eq:potential-interpolation}
\end{equation}
\end{lem}

\begin{proof}
Write $\delta=y-x$. Block smoothness gives 
\[
U_{0}\left(S_{m}\left(x,y\right)\right)\leq U_{0}\left(x\right)+t_{m}\left\langle \nabla_{m}U_{0}\left(x\right),\delta_{m}\right\rangle +\frac{L_{m}t^{2}_{m}}{2}\left\Vert \delta_{m}\right\Vert ^{2}.
\]
Multiplying by $\theta_{m}$, summing, and using $\theta_{m}t_{m}=\lambda$,
\begin{equation}
\sum^{M}_{m=1}\theta_{m}U_{0}\left(S_{m}\left(x,y\right)\right)\leq U_{0}\left(x\right)+\lambda\left\langle \nabla U_{0}\left(x\right),\delta\right\rangle +\frac{\lambda^{2}}{2}\left\Vert \delta\right\Vert ^{2}_{L,\theta}.\label{eq:potential-smooth-step}
\end{equation}
Selection-adapted convexity gives
\[
\lambda U_{0}\left(y\right)\geq\lambda U_{0}\left(x\right)+\lambda\left\langle \nabla U_{0}\left(x\right),\delta\right\rangle +\frac{\lambda^{2}}{2}\left\Vert \delta\right\Vert ^{2}_{L,\theta}.
\]
Comparison with \eqref{eq:potential-smooth-step} yields 
\begin{equation}
\sum^{M}_{m=1}\theta_{m}U_{0}\left(S_{m}\left(x,y\right)\right)\leq\left(1-\lambda\right)U_{0}\left(x\right)+\lambda U_{0}\left(y\right).\label{eq:potential-U0}
\end{equation}

For the separable part, convexity of $U_{j}$ gives 
\[
U_{j}\left(\left(1-t_{j}\right)x_{j}+t_{j}y_{j}\right)\leq\left(1-t_{j}\right)U_{j}\left(x_{j}\right)+t_{j}U_{j}\left(y_{j}\right).
\]
It follows that
\begin{align*}
 & \sum^{M}_{m=1}\theta_{m}\sum^{M}_{j=1}U_{j}\left(\left(S_{m}\left(x,y\right)\right)_{j}\right)\\
 & \quad=\sum^{M}_{j=1}\left[\left(1-\theta_{j}\right)U_{j}\left(x_{j}\right)+\theta_{j}U_{j}\left(\left(1-t_{j}\right)x_{j}+t_{j}y_{j}\right)\right]\\
 & \quad\leq\left(1-\lambda\right)\sum^{M}_{j=1}U_{j}\left(x_{j}\right)+\lambda\sum^{M}_{j=1}U_{j}\left(y_{j}\right).
\end{align*}
Adding this inequality to \eqref{eq:potential-U0} proves the claim. 
\end{proof}

\subsection{Proof of the Gibbs result}\label{subsec:gibbs-proof}

For a probability measure $\nu$ with Lebesgue density $\rho$, define
the potential energy and negative entropy 
\[
\mathcal{U}\left(\nu\right):=\int_{\mathbb{R}^{d}}U\,\mathrm{d}\nu,\qquad\mathcal{H}\left(\nu\right):=\int_{\mathbb{R}^{d}}\rho\log\rho\,\mathrm{d}x.
\]
Whenever the terms are finite, it holds that $\mathsf{KL}\left(\nu\mid\pi\right)=\mathcal{U}\left(\nu\right)+\mathcal{H}\left(\nu\right)$.
The entropy part of the proof is the following weighted form of the
triangular transport estimate in \cite{ascolani2026entropy}.
\begin{lem}[Weighted triangular entropy interpolation]
\label{lem:entropy} Let $\mu$ be absolutely continuous, let $T:\mathbb{R}^{d}\to\mathbb{R}^{d}$
be an increasing triangular $C^{1}$-diffeomorphism with respect to
a scalar ordering compatible with the block ordering, and let $0<\lambda\leq\theta_{\min}$.
Define 
\begin{equation}
\widehat{T}^{\,m}\left(x\right):=S_{m}\left(x,T\left(x\right)\right),\qquad t_{m}=\frac{\lambda}{\theta_{m}}.\label{eq:partial-transport}
\end{equation}
Then 
\begin{equation}
\sum^{M}_{m=1}\theta_{m}\mathcal{H}\left(\left(\widehat{T}^{\,m}\right)_{\#}\mu\right)\leq\left(1-\lambda\right)\mathcal{H}\left(\mu\right)+\lambda\mathcal{H}\left(T_{\#}\mu\right).\label{eq:entropy-interpolation}
\end{equation}
\end{lem}

\begin{proof}
The derivative $\mathrm{D}T\left(x\right)$ is lower triangular with
positive diagonal entries. Set $\mathbf{A}_{m}\left(x\right):=\left[\mathrm{D}T\left(x\right)\right]_{mm}$.
Each partial map $\widehat{T}^{\,m}$ is triangular and $\det\mathrm{D}\widehat{T}^{\,m}\left(x\right)=\det\left(\left(1-t_{m}\right)\mathbf{I}_{d_{m}}+t_{m}\mathbf{A}_{m}\left(x\right)\right)$.
The change-of-variables formula therefore gives 
\begin{align*}
\mathcal{H}\left(\left(\widehat{T}^{\,m}\right)_{\#}\mu\right) & =\mathcal{H}\left(\mu\right)-\int_{\mathbb{R}^{d}}\log\det\left(\left(1-t_{m}\right)\mathbf{I}_{d_{m}}+t_{m}\mathbf{A}_{m}\left(x\right)\right)\mu\left(\mathrm{d}x\right).
\end{align*}
For $a>0$ and $t\in\left[0,1\right]$, concavity of the logarithm
gives $\log\left(1-t+ta\right)\geq t\log a$. Applying this to the
diagonal entries of $\mathbf{A}_{m}\left(x\right)$, and using $\theta_{m}t_{m}=\lambda$,
yields that
\begin{align*}
 & \sum^{M}_{m=1}\theta_{m}\mathcal{H}\left(\left(\widehat{T}^{\,m}\right)_{\#}\mu\right)\\
 & \quad\leq\mathcal{H}\left(\mu\right)-\lambda\int_{\mathbb{R}^{d}}\sum^{M}_{m=1}\log\det\mathbf{A}_{m}\left(x\right)\mu\left(\mathrm{d}x\right)\\
 & \quad=\mathcal{H}\left(\mu\right)-\lambda\int_{\mathbb{R}^{d}}\log\det\mathrm{D}T\left(x\right)\mu\left(\mathrm{d}x\right)\\
 & \quad=\left(1-\lambda\right)\mathcal{H}\left(\mu\right)+\lambda\mathcal{H}\left(T_{\#}\mu\right),
\end{align*}
where triangularity gives $\det\mathrm{D}T=\prod_{m}\det\mathbf{A}_{m}$. 
\end{proof}

\begin{proof}[Proof of Theorem~\ref{thm:gibbs}]
Set $\lambda=\lambda^{\star}_{\theta}$. We first give the central
argument when the measures and maps are regular enough for the preceding
formulas to apply. Let $T$ be the increasing Knothe--Rosenblatt
transport from $\mu$ to $\pi$, so that $T_{\#}\mu=\pi$, and define
$\widehat{T}^{\,m}$ by \eqref{eq:partial-transport}.

Lemma~\ref{lem:potential}, applied pointwise with $y=T\left(x\right)$
and integrated against $\mu$, gives 
\begin{equation}
\sum^{M}_{m=1}\theta_{m}\mathcal{U}\left(\left(\widehat{T}^{\,m}\right)_{\#}\mu\right)\leq\left(1-\lambda\right)\mathcal{U}\left(\mu\right)+\lambda\mathcal{U}\left(\pi\right).\label{eq:gibbs-potential}
\end{equation}
Lemma~\ref{lem:entropy} gives the same inequality for $\mathcal{H}$.
Adding the two and decomposing the $\mathsf{KL}$ term gives 
\begin{equation}
\sum^{M}_{m=1}\theta_{m}\mathsf{KL}\left(\left(\widehat{T}^{\,m}\right)_{\#}\mu\mid\pi\right)\leq\left(1-\lambda\right)\mathsf{KL}\left(\mu\mid\pi\right),\label{eq:gibbs-transport-bound}
\end{equation}
because $\mathsf{KL}\left(\pi\mid\pi\right)=0$.

The map $\widehat{T}^{\,m}$ leaves $x_{-m}$ unchanged. Hence $\left(\widehat{T}^{\,m}\right)_{\#}\mu$
and $\mu$ have the same $(-m)$-marginal, and \eqref{eq:gibbs-variational}
implies that $\mathsf{KL}\left(\mu_{-m}\mid\pi_{-m}\right)\leq\mathsf{KL}\left(\left(\widehat{T}^{\,m}\right)_{\#}\mu\mid\pi\right)$.
Multiplying by $\theta_{m}$, summing, and applying \eqref{eq:gibbs-transport-bound}
proves \eqref{eq:weighted-tensorization}. The chain rule 
\begin{align*}
\mathsf{KL}\left(\mu\mid\pi\right) & =\mathsf{KL}\left(\mu_{-m}\mid\pi_{-m}\right)+\mathbf{E}_{\mu_{-m}}\left[\mathsf{KL}\left(\mu\left(\mathord{\cdot}\mid X_{-m}\right)\mathrel{\big|}\pi\left(\mathord{\cdot}\mid X_{-m}\right)\right)\right]
\end{align*}
then gives \eqref{eq:weighted-conditional}.

Finally, convexity of relative entropy and \eqref{eq:gibbs-variational}
yield 
\begin{align*}
\mathsf{KL}\left(\mu P^{\theta}\mid\pi\right) & \leq\sum^{M}_{m=1}\theta_{m}\mathsf{KL}\left(\mu P_{m}\mid\pi\right)\\
 & =\sum^{M}_{m=1}\theta_{m}\mathsf{KL}\left(\mu_{-m}\mid\pi_{-m}\right)\\
 & \leq\left(1-\lambda\right)\mathsf{KL}\left(\mu\mid\pi\right),
\end{align*}
which is \eqref{eq:gibbs-contraction}. As in \cite[Section~5.4]{ascolani2026entropy},
truncation and regularisation remove the temporary smoothness and
diffeomorphism assumptions; those approximation steps are unchanged
by the weights. 
\end{proof}

\subsection{Proof of the CAVI result}\label{subsec:cavi-proof}
\begin{proof}[Proof of Theorem~\ref{thm:cavi}]
Set $\lambda=\lambda^{\star}_{\theta}$, and fix $q=q_{1}\otimes\cdots\otimes q_{M}$
with $F\left(q\right)<+\infty$. For each $m$, let $T_{m}$ be the
quadratic optimal transport from $q_{m}$ to $q^{\star}_{m}$, and
write $T\left(x\right)=\left(T_{1}\left(x_{1}\right),\ldots,T_{M}\left(x_{M}\right)\right)$,
whereby $T_{\#}q=q^{\star}$.

Define $x^{\left[m\right]}:=S_{m}\left(x,T\left(x\right)\right)$,
$t_{m}=\frac{\lambda}{\theta_{m}}$, and let $q^{\left[m\right]}$
be the law of $x^{\left[m\right]}$ under $x\sim q$. This measure
differs from $q$ only in factor $m$, whose law lies at time $t_{m}$
on the quadratic Wasserstein geodesic from $q_{m}$ to $q^{\star}_{m}$.
By the variational definition of the CAVI update, 
\begin{equation}
F\left(\mathcal{C}_{m}q\right)\leq F\left(q^{\left[m\right]}\right).\label{eq:cavi-variational-bound}
\end{equation}

Lemma~\ref{lem:potential}, integrated with respect to $q$, gives
\begin{equation}
\sum^{M}_{m=1}\theta_{m}\int_{\mathcal{X}}U_{0}\,\mathrm{d}q^{\left[m\right]}\leq\left(1-\lambda\right)\int_{\mathcal{X}}U_{0}\,\mathrm{d}q+\lambda\int_{\mathcal{X}}U_{0}\,\mathrm{d}q^{\star}.\label{eq:cavi-potential}
\end{equation}
We used only the $U_{0}$ part of the lemma here; it is convenient
to combine the separable potentials with entropy. For a probability
measure $\nu$ on $\mathcal{X}_{m}$, define $\mathcal{G}_{m}\left(\nu\right):=\int_{\mathcal{X}_{m}}U_{m}\,\mathrm{d}\nu+\mathcal{H}\left(\nu\right)$.
Since $U_{m}$ is convex and negative entropy is displacement convex,
$\mathcal{G}_{m}$ is convex along quadratic Wasserstein geodesics.
If $q_{m,t}$ is the geodesic from $q_{m}$ to $q^{\star}_{m}$, then
\begin{equation}
\mathcal{G}_{m}\left(q_{m,t_{m}}\right)\leq\left(1-t_{m}\right)\mathcal{G}_{m}\left(q_{m}\right)+t_{m}\mathcal{G}_{m}\left(q^{\star}_{m}\right).\label{eq:Gm-convexity}
\end{equation}
Only factor $m$ changes in $q^{\left[m\right]}$. Therefore 
\begin{align}
 & \sum^{M}_{m=1}\theta_{m}\sum^{M}_{j=1}\mathcal{G}_{j}\left(\left(q^{\left[m\right]}\right)_{j}\right)\nonumber \\
 & \quad=\sum^{M}_{j=1}\left[\left(1-\theta_{j}\right)\mathcal{G}_{j}\left(q_{j}\right)+\theta_{j}\mathcal{G}_{j}\left(q_{j,t_{j}}\right)\right]\nonumber \\
 & \quad\leq\left(1-\lambda\right)\sum^{M}_{j=1}\mathcal{G}_{j}\left(q_{j}\right)+\lambda\sum^{M}_{j=1}\mathcal{G}_{j}\left(q^{\star}_{j}\right).\label{eq:cavi-separable}
\end{align}
For product measures, 
\[
F\left(q\right)=\int_{\mathcal{X}}U_{0}\,\mathrm{d}q+\sum^{M}_{m=1}\mathcal{G}_{m}\left(q_{m}\right).
\]
Adding \eqref{eq:cavi-potential} and \eqref{eq:cavi-separable} gives
\begin{equation}
\sum^{M}_{m=1}\theta_{m}F\left(q^{\left[m\right]}\right)\leq\left(1-\lambda\right)F\left(q\right)+\lambda F\left(q^{\star}\right).\label{eq:cavi-comparison}
\end{equation}
Together with \eqref{eq:cavi-variational-bound}, this proves \eqref{eq:cavi-one-step}.
Conditional expectation gives \eqref{eq:cavi-conditional}, and the
tower property then yields \eqref{eq:cavi-iterate}. 
\end{proof}

\section{Discussion}\label{sec:discussion}

An immediate question is how to choose $\theta$. When a fixed lower
curvature matrix is available, computable lower bounds on $\lambda^{\star}_{\theta}$
can be obtained from eigenvalue-type calculations. More generally,
the exact selection-adapted convexity constant may be unavailable,
and one may instead optimise a computable lower bound obtained from
more tractable global and blockwise curvature estimates. Such a procedure
optimises our bound on the contraction rate rather than necessarily
the exact convergence rate. It would therefore be useful to examine
the stability of the resulting scan under errors in the curvature
quantities and to develop robust choices when only bounds or estimates
are available.

The weighted geometry encoded by $\left\Vert \cdot\right\Vert _{L,\theta}$
gives some indication as to the active ingredients in this optimisation
problem, and what phenomena one seeks to balance. Increasing $\theta_{m}$
relaxes the relative convexity requirement for the $m$th block, while
implicitly making the requirement on other blocks more stringent.
An effective scan therefore assigns effort to the blocks which participate
most strongly in the `least convex' directions, while retaining enough
probability on the remaining blocks that the coordinate-wise upper
bound $\lambda^{\star}_{\theta}\leq\theta_{\min}$ does not become
limiting. In effect, then, scan choice can often be understood less
as updating `rougher' blocks more frequently, and more as a mechanism
for effectively exposing convexity and suppressing strong cross-block
dependence (which are often two sides of the same coin).

Scan design should be distinguished from the choice of coordinates
and blocking structure. An invertible transformation within each block
leaves the Gibbs sampler and CAVI equivalent, up to applying the inverse
transformation, but the scalar convexity and smoothness constants
used in a complexity certificate are \emph{not} invariant under such
a transformation. Blockwise reparameterisation can therefore sharpen
the available bound without changing the underlying algorithm in an
essential way. This contrasts with the non-uniform selection mechanism,
which genuinely changes the law of the iterates themselves. 

More structural choices, including blocking, centering, non-centering,
and partial centering, may change both the dependence structure and
the cost of a coordinate update. For CAVI, they may also change the
mean-field family and hence the approximation being sought. These
choices are complementary to scan adaptation: one may first choose
a representation which exposes weak dependence and tractable updates,
and then select $\theta$ to account for the residual heterogeneity
across the resulting blocks. In regimes where a strongly coupled collection
of variables can be updated jointly at reasonable cost, changing the
blocking may be more consequential than changing the scan probabilities.

Update costs can also vary considerably across blocks. If block $m$
costs $c_{m}$ to update, then maximising $\lambda^{\star}_{\theta}$
per iteration need not be optimal from the perspective of convergence
per unit of computation. A natural fixed-scan proxy for the contraction
rate per expected unit cost is
\[
\frac{-\log\left(1-\lambda^{\star}_{\theta}\right)}{\sum^{M}_{m=1}\theta_{m}c_{m}}\sim\frac{\lambda^{\star}_{\theta}}{\sum^{M}_{m=1}\theta_{m}c_{m}}\qquad\text{as }\lambda^{\star}_{\theta}\to0.
\]
This criterion is only a first approximation because computational
costs may depend on the implementation and the current state, but
it makes clear that the preferred scan can depend on whether convergence
is measured per iteration or per unit computation.

Taken together, these considerations position $\lambda^{\star}_{\theta}$
as a principled fixed-scan design criterion, rather than a universal
notion of optimality. Once the parameterisation and blocking structure
have been fixed, it provides a common way to relate the residual block
geometry to the allocation of updates in both Gibbs sampling and CAVI.
A scan chosen in this way should therefore be understood as optimising
a concrete, theoretically grounded criterion, subject to the separate
trade-offs created by computational cost and uncertainty in the curvature
information.

\section{Acknowledgements}\label{sec:acknowledgments}

During the preparation of this work, the author used OpenAI ChatGPT
and OpenAI Codex to explore proof strategies and assist with the drafting
and revision of the manuscript. In particular, discussions with ChatGPT
suggested the weighted adaptation of the coordinate interpolation
argument developed in Section~\ref{subsec:interpolation}, building
on the argument in \cite{ascolani2026entropy}. The author independently
checked the resulting arguments and takes full responsibility for
the contents of the paper.

An initial draft of this work treated only the Gibbs sampler; feedback
from Giacomo Zanella motivated the author to explore the extension
to CAVI. Additional feedback from Filippo Ascolani and Hugo Lavenant
was useful in guiding which secondary results to include and highlight.
All of these inputs are gratefully acknowledged.

\bibliographystyle{plain}
\bibliography{selection_adapted_convexity}

\end{document}